\documentclass[conference,twocolum,final]{IEEEtran}

\usepackage{bbm}
\usepackage{algorithmic}
\usepackage{algorithm} 
\usepackage{amsthm}
\usepackage{amsmath,amssymb}
\usepackage{subfigure}
\ifCLASSINFOpdf 
\usepackage[pdftex]{graphicx} 
\usepackage[pdftex,colorlinks,
           bookmarks=true,%
            pdftitle=Optimal\ Geographic\ Caching\ In\ Cellular Networks,%
           pdfauthor=B.\ Blaszczyszyn\ -\ A.\ Giovanidis]{hyperref}  
\usepackage[numbers,sort&compress]{natbib} 
\usepackage{hypernat} 
\else

\usepackage[dvips]{graphicx} 
\usepackage[dvips,colorlinks,
           bookmarks=true,%
           pdftitle=Optimal\ Geographic\ Caching\ In\ Cellular Networks,%
           pdfauthor=B.\ Blaszczyszyn\ -\ A.\ Giovanidis]{hyperref}  
 \usepackage[numbers,sort&compress]{natbib} 
\fi 
\hypersetup{
   bookmarksnumbered,
   pdfstartview={FitH},
   citecolor={blue},
   linkcolor={red},
   urlcolor=[rgb]{0,0.55,0},
   pdfpagemode={UseOutlines}}

\newtheorem{Exm}{Example}

\newtheorem{Lem}{Lemma}

\newtheorem{Fact}{Fact}

\newtheorem{Theorem}{Theorem}
\newtheorem*{proof*}{Proof}

\begin{document}

\title{Optimal Geographic Caching In Cellular Networks}

\author{Bart{\l}omiej~B{\l}aszczyszyn$^\dagger$ and Anastasios~Giovanidis$^\ast$\\[2ex]}

\maketitle

\begin{abstract}
In this work we consider the problem of an optimal geographic placement of content  in wireless cellular networks modelled by Poisson point processes. Specifically, for the typical user requesting some particular content and whose popularity follows a given law (e.g. Zipf),  we calculate the probability of finding the content cached in one of the base stations. Wireless coverage follows the usual signal-to-interference-and noise ratio (SINR) model, or some variants of it.  We formulate and solve the problem of an optimal randomized content placement policy, to maximize the user's hit probability. The result dictates that it is not always optimal to follow the standard policy "cache the most popular content, everywhere". In fact, our numerical results regarding three different coverage scenarios, show that  
the optimal policy significantly increases the chances of hit under high-coverage regime, i.e., when  the probabilities of coverage by more than just one station are high enough.
\end{abstract}

\begin{keywords}
wireless cache; Poisson cellular network; SINR k-coverage; hit probability; content popularity; optimization \end{keywords}

\let\thefootnote\relax\footnotetext{\hspace{-2ex}$^\dagger$Inria/Ens, 23
  av. d'Italie 75214 Paris, France; Bartek.Blaszczyszyn@ens.fr\\ 
$^\ast$CNRS - T\'el\'ecom ParisTech,
23 Avenue d'Italie,
75013 Paris, France; anastasios.giovanidis@telecom-paristech.fr}
\newcommand{\thefootnote}{\arabic{footnote}}

\section{Introduction}

Today's cellular networks provide additionally to 
traditional telephony and messaging services, a considerable amount of multimedia content. Multimedia traffic demand is expected to show exponential increase in the years to come. Recent suggestions to densify the network via multi-tier heterogeneous equipment 
or to apply cooperation techniques \cite{BacGiov13}, will locally improve wireless throughput but will eventually push the network's backhaul and its available bandwidth to its limit. 

Considering the fact that the great volume of traffic consists of multiple demands for the same content by various users, a solution to relieve the overloaded network is to cache popular content at intermediate nodes. In the case of cellular networks, this practically translates in adding physical memory at the central base station (BS) and the smaller scale (pico, nano, etc.) stations. 
There are several benefits from doing so: (a) The most evident is the reduction of backhaul traffic load. (b) Another one is the reduction of multimedia (audio/video) playback latencies. When the content is cached at a node close to the user, it is delivered with less delay than fetching it from the core network. (c) Caching can give the opportunity to adapt the multimedia quality to the actual end-users' channel, and consequently improve Quality-of-Experience \cite{PoularINFOCOM14}.

In the literature there are already considerable works dealing with the problem of cellular caching. Among these, the paper by  Shanmugam et al \cite{DimakisCaireCacheIT} treats the problem of optimal association of content to wireless caches, given BS-user topology. Bastug et al in \cite{BastugCacheISWCS14} provide stochastic geometry results on the user outage probability and average delay experienced. The optimal storage allocation related to user mobility is addressed by Poularakis and Tassiulas in \cite{PoularISIT13}. Further benefits and challenges from the application of caching in 5G networks are presented in the work by Wang et al \cite{Wang5Gcache}. In all the above, it is common to consider a fixed content library and a popularity distribution which is known a priori and follows the Zipf distribution. The latter was proposed and verified to well approximate the hit distribution of Internet content by Newman \cite{Newman05}.

The great difference of caching in cellular networks as compared to wired ones, is that there can appear planar regions with overlapping coverage by more than one BS \cite{KeelerBartek13}. When a user finds him/herself in such areas, he/she can choose service by any of the covering stations. Such observation has indeed been taken into consideration in the works of \cite{DimakisCaireCacheIT}, \cite{PoularINFOCOM14}, \cite{PoularISIT13}, but the discrete problem formulations suggested fail to give solutions with global validity, because they are based on a priori known BS-user topologies. 

In our work we will revisit the problem of optimal content placement in cellular caches by assuming a known distribution of the \textit{coverage number}, i.e. the number of BSs simultaneously covering a user. We will see different expressions for this distribution based on different coverage models, like the $\mathrm{SINR}$ model in \cite{KeelerBartek13} where the network topology was modelled by a Poisson point process. As a main result, we will provide the \textit{optimal probabilistic placement policy}, which guarantees maximal total hit probability for random network topologies. To achieve optimality, the policy exploits multi-coverage regions and delivers considerable performance improvement compared to the standard "cache the most popular content, everywhere" strategy. We argue that the latter is not optimal in general networks but only either for an isolated cache \cite{LeonardiCacheUni} or when there is no coverage overlap in the network. 

The model under study is presented in Section \ref{secII}, where the probabilistic caching policy is introduced and three coverage model examples are given. In Section \ref{secIII} we state and solve the content placement optimization problem. Evaluation of the optimal strategy for three different coverage models is presented in Section \ref{secV}. We conclude the paper in Section~\ref{secVI}.

\section{Model under study}
\label{secII}

We consider a cellular network where the positions of Base Stations (BSs) coincide with the atoms from the realization of a two dimensional (2D) Poisson Point Process (PPP) $\Phi=\left\{x_i\right\}$. The PPP is homogeneous and has intensity $\lambda>0$.
The performance of the network is evaluated at the Cartesian origin $\left(0,0\right)$, which we denote as the \textit{typical user} $o$. Due to the Slivnyak-Mecke theorem and the stationarity and isotropy of the PPP, the results for the typical user apply to any user randomly located on the 2D plane \cite{BacBlaVol1}. 

\subsection{Network Coverage Number}
In cellular networks a user at a random location may be covered by more than one BS, or may not be covered at all. A user may also be covered by multiple networks. The so called \textit{coverage number} $\mathcal{N}$ is a random variable (r.v.) \cite{KeelerBartek13} that depends on the features of the communications scheme and the network parameters. It has a mass function
\begin{equation}
\label{CovNum}
p_m:=\mathbb{P}\left[\mathcal{N}=m\right],\qquad m=0,1,\ldots
\end{equation}
The maximum number of covering BSs is
$M\in\mathbb{N}^+\cup\left\{\infty\right\}$, where we let the coverage
number be unlimited in the general case. Later in the section, we will
give specific expressions for the $p_m$, which depend on the network
evaluated.  Obviously, it holds
\begin{equation}
\label{SumPm}
\sum_{m=0}^{M} p_m = 1.
\end{equation}
Some specific coverage models will be presented in Section~\ref{ss.CoverageModels}.

\subsection{Content and its Popularity}
Each user has a request for a specific content (say video file) that he/she wants to receive. In our model the set (library) of available content is finite. It is denoted by $\mathcal{C}:=\left\{c_1,c_2,\ldots,c_J\right\}$, where an element $c_j$ is an entire file. The cardinality of the set is $J$. We consider that all content has the same size, normalised to $1$.  Cases of unequal size will not be treated in this work, but we can always assume that each file can be divided into chunks of equal size, so the same analysis can still be applied. 
Furthermore, each content is related to its popularity, which we assume known a priori. We order the content by popularity: $c_1$ is the most popular content, $c_2$ the second most popular and so on. The popularity follows a distribution $\left\{a_j\right\}$. To be consistent with the above ordering, $a_1\geq a_2\geq \ldots\geq a_J$. Without losing in generality, we will often consider that the distribution has a Zipf probability mass function and consequently the probability that a user (hence the typical one) will ask for content $c_j$ is equal to 
\begin{equation}
\label{PMFzipf}
a_j  =  A^{-1}j^{-\gamma}.
\end{equation}
Here, $\gamma$ is the Zipf exponent, often (but not necessarily) chosen as $\gamma<1$, so that $a_1/a_2=2^{\gamma}<2$. 
It holds 
\begin{equation}
\label{sumAJ}
\sum_{j=1}^{J}a_j  =  1,
\end{equation}
and this explains the normalisation factor $A:=\sum_{j=1}^{J} j^{-\gamma}$. 


\subsection{Content Placement to Caches}
We assume that a cache memory of size $K\geq 1$ is installed and available on each BS. The memory inventory of BS $x_i\in\Phi$ is denoted by $\Xi^{(i)}$,  which is a subset of  $\mathcal{C}$, 
with the number of elements not larger than $K$; i.e., $|\Xi^{(i)}|\le K$ for all $i$.

We consider a probabilistic model, where the content is \textit{independently},  placed in the cache memories of different BSs, according to {\em the same distribution}. In other words, $\Xi^{(i)}$ are assumed independent identically distributed (random) subsets of~$\mathcal{C}$.
Denote by 
$$b_j:=\mathbf{P}\left(c_j\in\Xi^{(i)}\right)$$
the probability that the content $c_j$ is stored at a given base station.
(The model is homogeneous, hence the values of $b_j$ are common for all BSs $x_i$;  and the superscript ${\cdot}^{(i)}$ can be omitted when considering a generic base station.) 

As we shall see, the probability that a typical user finds the content he/she is looking for in the inventory of a  base station covering him/her (which is the performance metric we want to maximize) depends on the distribution of the (random) set $\Xi$ {\em only} through the one-set-coverage probabilities $b_j$, $j=1,\ldots,J$. And these probabilities do not define the distribution of the random set 
$\Xi$ --- hence the content placement policy ---  uniquely.

When looking for the optimal values $b_j$, $j=1,\ldots,J$, we shall consider the following constraints
\begin{eqnarray}
\label{sumBKa}
\sum_{j=1}^{J}b_j \leq K, & \\
\label{sumBKb}
0\leq b_j\leq 1, & \forall j.
\end{eqnarray}
The second condition is obvious ($b_j$ is a probability). Regarding the first one we have the following result.
\begin{Fact}
Assuming~(\ref{sumBKb}), the condition~(\ref{sumBKa}) is necessary and sufficient for the existence 
of a distribution of $\Xi$ satisfying   $|\Xi|\le K$ almost surely, i.e. existence of a random content placement policy requiring no  more than $K$ slots of  memory at each base station.
\end{Fact}
\begin{IEEEproof}
The necessity follows from the observation that the right-hand side of~(\ref{sumBKa}) is equal to the 
expected number of content items in the base station inventory. Indeed
$$\mathbb{E}\left[\sum_{j=1}^J \mathbbm{1}(c_j\in\Xi)\right]=\sum_{j=1}^J\mathbf{P}\left(c_j\in\Xi\right)=\sum_{j=1}^{J}b_j\,.$$
We prove  the sufficiency by constructing (in what follows) some particular content placement policy satisfying  $|\Xi|\le K$.
\end{IEEEproof}

\begin{figure}[t!]
\centering
\includegraphics[trim = 0mm 0mm 0mm 0mm, clip, width=0.20\textwidth]{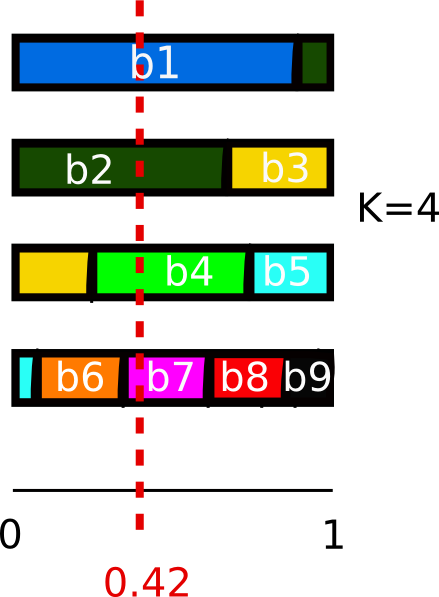}
\caption{A realization of the probabilistic placement policy for the case of $J=9$ contents and $K=4$ memory slots. We first draw uniformly a random number (0.42). The vertical line at this point intersects with each of the $4$ memory intervals at a specific content. We conclude from the figure that the subset $\left\{c_1,c_2,c_4,c_7\right\}$ will be cached.}
\label{PlacePolicy}
\end{figure}

The following policy having one-set-coverage probabilities $b_j$ satisfying~(\ref{sumBKa}) and (\ref{sumBKb}) respects the cache size constraint $K$.

\textbf{Probabilistic placement policy:} Given the cache memory of size $K\geq 1$, and 
 the values $b_j$, $j=1,\ldots,J$ satisfying~(\ref{sumBKa}) and (\ref{sumBKb}), 
we divide it into $K$ continuous memory intervals of unit length and place them one under the other, as shown in the example considered on Fig. \ref{PlacePolicy} (the example assumes  equality in~(\ref{sumBKa})). The $J$ contents of the library are picked one after the other without replacement and their values $b_j$ fill the memory. If not enough space is available in one unitary memory slot, the content fills the slot underneath. 
In order to randomly choose a set of contents, 
we pick uniformly a number within $\left[0,1\right]$ and draw a vertical line which intersects the memory space covered by no more than $K$ distinct contents, (It intersects exactly $K$ contents if the equality is observed in~(\ref{sumBKa}).) The contents are distinct because $b_j\in\left[0,1\right]$. 
Moreover, the probability of appearance of content $j$ in a memory of size $K$ is exactly equal to $b_j$. 


\section{Optimal Content Placement --- Problem Statement and Solution}
\label{secIII}

The performance metric of interest is the \textit{total hit probability}, i.e. the probability that the typical user will find the content he/she is asking for in one of the BSs he/she is covered from. This is $1$-minus the probability that the user does not find its content. This happens when the user is covered by $m=0$ BSs (i.e. no coverage), or by some number $m>0$ but the content has not been saved in the cache memory space of any of these BSs. The performance metric is equal to
\begin{equation}
\label{PerfMe}
f\left(b_1,\ldots,b_J\right) := 1-\sum_{j=1}^{J}a_j  \sum_{m=0}^{\infty} p_m\left(1-b_j\right)^{m}.
\end{equation}
To better understand the above expression, $m$ is the number of BSs that the user is covered by. The probability $\left(1-b_j\right)^{m}$ is the probability that none of the cache memory slots of these stations contains the desired content. Hence, the expression in (\ref{PerfMe}) is the probability that the content requested by the user should be fetched from the backhaul network.

We can control the hit probability, by varying the content placement probabilities $b_j$. In the following we will find the optimal vector $\left(b_1,\ldots,b_J\right)$, that maximises the objective function $f$ in (\ref{PerfMe}). The constraint set of our problem is 
\begin{eqnarray}
\label{CS1}
\mathcal{F}_1 & := & \left\{ (b_1,\ldots,b_J)|\ b_1+\ldots+b_J\leq K,\right.\nonumber\\
& & \left. \ \& \ b_j\in\left[0,1\right], \ \forall j\right\}
\end{eqnarray}
The parameters that influence the objective function but cannot be used as actions are (a) the size of memory $K$, (b) the probability of content popularity $a_j$, $j=1,\ldots,J$ and (c) the probability of coverage by m BSs $p_m$, $m=0,1,\ldots$. 
%
%
%
Altogether, we summarize the optimisation problem below, which we call GCP (Geographic Caching Problem) and in the following we will try to find its solution. 
\begin{center}
\begin{tabular}{l l c}
$\mathbf{\max}$ 	& 	$f\left(b_1,\ldots, b_J\right)$	& \textbf{[GCP]}\\
\textbf{s.t.}		&	$\left(b_1,\ldots,b_J\right)\in\mathcal{F}_1$		&
\end{tabular}.
\end{center}
We first give two Lemmas that facilitate the solution.
\begin{Lem}
\label{Lem1}
The objective function of [GCP] has the following two properties:
\begin{itemize}
\item \textbf{P.1:} It is separable w.r.t. 
$b_1,\ldots,b_J$.
%
\item \textbf{P.2:} It is increasing and concave in $b_j$, $\forall j$. Consequently, it is a concave function of $\left(b_1,\ldots,b_J\right)$.
\end{itemize}
\end{Lem}
\begin{proof}
\textbf{P.1} comes by rewriting the objective function after replacing $1=\sum_{j=1}^Ja_j$ (see also (\ref{sumAJ})). \textbf{P.2} becomes trivial due to the separability property. We only need to show that the first derivative of $g(b_j)$ is $\geq 0$ and the second $\leq 0$, $\forall b_j$.
\end{proof}
\begin{Lem}
\label{Lem2}
At the optimal solution, the sum constraint inequality (\ref{sumBKa}) is active, i.e. the optimal solution satisfies 
\begin{equation}
\label{EqConst}
b_1^*+\ldots+b_J^*  =  K.
\end{equation}
\end{Lem}
\begin{proof} Suppose that the inequality is inactive, i.e. strictly $<K$ for the optimal solution. But then, for any $l\leq J$, we can increase $b_l^*\rightarrow b_l^*+\epsilon$, so that the constraint is satisfied with equality. Substituting in the objective function $b_l^*+\epsilon$ instead of $b_l^*$ the value of the function will increase, because $f$ is increasing over $b_l$, by \textbf{P.2}. Hence the primal optimal solution cannot leave the constraint active and (\ref{EqConst}) is true.
\end{proof}
Since the objective function is concave by \textbf{P.2} and the constraint set is linear (affine inequalities), the optimisation problem can be solved as a \textit{convex program}. We will make use of the Lagrangian relaxation method (see \cite{BoydBook}). 
Let us relate the dual price $\mu\geq 0$ to the sum constraint inequality (\ref{sumBKa}). 
The Lagrangian function is 
\begin{eqnarray}
\label{LagPerfMe}
L\left(b_1,\ldots,b_J,\mu\right) & = & \sum_{j=1}^Ja_j\left(1-\sum_{m=0}^{\infty}p_m\left(1-b_j\right)^{m}\right)+\nonumber\\
& +& \mu\left(K-\sum_{j=1}^Jb_j\right),
\end{eqnarray}
and the remaining constraint set is 
$\mathcal{F}_2:=\left\{ b_j\in\left[0,1\right], \ \forall j=1,\ldots,J\right\}$.
%
We can systematically find the optimal primal $(b_j^*)$ and dual ($\mu^*$) variables by solving a min-max problem. Additionally, in our case where we deal with a convex program, the optimal value of the min-max problem ($f^*$) is equal to the optimal value of the original problem [GCP] with objective function $f$, that is
\begin{eqnarray}
\label{MinMaxL}
f^*:=\max_{\mathcal{F}_1} f\left(b_1,\ldots,b_J\right) & = & \min_{\mu\geq 0}\max_{\mathcal{F}_2}L\left(b_1,\ldots,b_J,\mu\right).\nonumber\\
 & = & f\left(b_1^*,\ldots,b_J^*\right).
\end{eqnarray}
We then say that the duality gap between the original [GCP] problem and the min-max problem is zero. 
\begin{Theorem}
\label{Th1}
The optimal primary variables, that maximise the [GCP] objective, given the optimal dual variable $\mu^*$ is $b_j^*=b_j\left(\mu^*\right)$ with the expression
\begin{eqnarray}
\label{SolveSubJJ}
b_j\left(\mu^*\right) & = & \left\{
\begin{tabular}{l l}
$1$, & if $a_j p_1> \mu^*$\\
$\omega\left(\mu^*\right)$, & if $a_j p_1 \leq \mu^* \leq a_j\mathbb{E}\left[\mathcal{N}\right]$\\
$0$, & if $a_j\mathbb{E}\left[\mathcal{N}\right]< \mu^*$ 
\end{tabular}.
\right.
\end{eqnarray}
In the above, $\mathbb{E}\left[\mathcal{N}\right]=\sum_{m=1}^{\infty}m p_m$ and $\omega\left(\mu^*\right)$ is the solution over $b_j$ of the equation
\begin{equation}
\label{Solve2}
a_j\sum_{m=1}^M p_m m (1-b_j)^{m-1}=\mu^*.
\end{equation}
The optimal dual variable $\mu^*$ satisfies the equality
\begin{equation}
\label{Solve3}
b_1\left(\mu^*\right)+\ldots+b_J\left(\mu^*\right)  =  K.
\end{equation}
\end{Theorem}

\begin{proof}
\textbf{Sketch.} Given any dual price $\mu\geq 0$, we first solve the relaxed primary problem $\max_{\mathcal{F}_2}L\left(b_1,\ldots,b_J,\mu\right)$ over the $b_j$s as shown in (\ref{MinMaxL}). From Lemma \ref{Lem1} and the affinity of the relaxed constraint, the primal problem is separable in $J$ subproblems, each one having as constraint $b_j\in\left[0,1\right]$. The solution for each $\mu$, hence also for $\mu^*$ is given in (\ref{SolveSubJJ}) and (\ref{Solve2}). We further need to minimize the function $q\left(\mu\right)=\max_{\mathcal{F}_2}L\left(b_1,\ldots,b_J,\mu\right)$ over $\mu$. The standard way to do this is by use of a subgradient method. However, in our problem we need not proceed this way due to Lemma \ref{Lem2}, which states that the optimal primal solution satisfies the relaxed constraint with equality. To find $\mu^*$, we thus have to replace the solution (\ref{SolveSubJJ}) in the equation $\sum_{j=1}^Jb_j\left(\mu\right)=K$ and solve over $\mu$. The solution is unique since we can prove that the sum of $b_j$s is a decreasing function in $\mu$ within an interval that is guaranteed to contain the solution.
\end{proof}
\textbf{Algorithm.} The solution is found numerically as follows. We start by an interval of $\mu$ that contains the optimal solution, i.e. $\mu^{(0)}\in\left[\mu^{(0,\min)},\mu^{(0,\max)}\right]=\left[a_Kp_1,a_1\mathbb{E}\left[\mathcal{N}\right]\right]$. Then we use the \textit{bisection method}, according to 
which $\sum_{j=1}^J b_j\left(\mu^{(l+1)}\right)$ is evaluated for $l=0,1,\ldots$ at the dual price 
$\mu^{(l+1)}:=\mu^{(l,\min)}+(\mu^{(l,\max)}-\mu^{(l,\min)})/2$.
If the value of the sum is $<K$ then the search continues to the left interval and $\mu^{(l+1,\max)}:=\mu^{(l+1)}$, else if the sum is $>K$, the search continues to the right interval and $\mu^{(l+1,\min)}:=\mu^{(l+1)}$. The algorithm stops when the change in $\mu$ for some step $l$ is smaller than a chosen $\epsilon>0$. The difficulty in the implementation lies in the fact that we need to solve also the polynomial equalities of the form (\ref{Solve2}), which do not give a closed form solution when $M$ is large. We also solve these over $b_j$ by use of the bisection method.

To provide some intuition, we consider in what follows a simple example assuming $J=2, K=1$.
\begin{Exm}[2CP]\label{Cor1}
Consider the case of one-slot cache memory $K=1$ and the  content library of size $J=2$ with $a_1+a_2=1$ and maximum $M=2$ BSs covering a user. We call this problem the [2CP]. It has the following explicit solution.
The optimal pair $(b_1^*,b_2^*)$ that solves the [2CP] problem is $b_2^*=1-b_1^*$ where
\begin{eqnarray}
\label{OFD2++}
b_1^* =\left\{\begin{tabular}{l l}
$\frac{2a_1\left(p_1+p_2\right)-p_1}{2p_2}$, & if $0 \leq a_1\leq 1-\frac{p_1}{2(p_1+p_2)}$.\\
$1$, & otherwise
\end{tabular}\right.
\end{eqnarray}
\end{Exm}

\section{Performance Evaluation}
\label{secV}

We will now evaluate the performance of the placement policy given in Theorem \ref{Th1} 
for some  three different coverage models, which we present first.

\subsection{Coverage Models}
\label{ss.CoverageModels}
In the following we will overview three specific network models that give different expressions to the coverage number probability $p_m$.

\subsubsection{SINR Model} The quality of coverage at the origin is described by the $\mathrm{SINR}_o$ (from now on $\mathrm{SINR}$). $\mathrm{SINR}(x_i)$ is the $\mathrm{SINR}$ at the reception, when user $o$ is connected to BS $x_i\in\Phi$ and is defined as
\begin{equation}
\label{SINRdef}
\mathrm{SINR}(x_i):=\frac{S_{i}/\ell(r_i)}{W+I-S_i/\ell(r_i)}.
\end{equation}
In the above, $S_i$ is the shadowing experienced between the typical user and the BS at $x_i$. The constant $W$ is the noise power, $I=\sum_{x_i\in\Phi}S_i/\ell(r_i)$ is the total received power from the network, $r_i=\left| x_i\right|$ is the distance of $x_i$ from $o$, and $\ell(r)=\left(Br\right)^{\beta}$ is the path-loss function, with constants $B>0$, $\beta>2$. We say that the typical user is covered when $\mathrm{SINR}(x_i)>T$, where $T$ is a predefined positive threshold.

The coverage number $\mathcal{N}\left(T\right)$ indicates how many BSs cover the typical user simultaneously and is the r.v.
\begin{equation}
\label{NT}
\mathcal{N}\left(T\right)  =  \sum_{x_i\in\Phi}\mathbf{1}[\mathrm{SINR}(x_i)>T].
\end{equation}
For the coverage of a user who can choose to be served by different BSs  in some realization of $\Phi$, we make use of a basic result from \cite[Proposition 6.2]{BacBlaVol1}, analysed further in \cite{KeelerBartek13}. 
It is shown that if $m$ stations cover a user at $\mathrm{SINR}$ level $T$, then the following inequality holds
\begin{equation}
\label{SINRT}
M= \left \lceil{\frac{1}{T}}\right \rceil  \Leftrightarrow  m< 1+1/T.
\end{equation}
For example, when $T\geq 1$ then necessarily $m< 1+\frac{1}{T}<2$ which implies that $m\in\left\{0,1\right\}$ and $M=1$. Similarly, when $1> T\geq1/2$,  $m\in\left\{0,1,2\right\}$ and $M=2$, etc. 
Based on this model, the authors in \cite{KeelerBartek13} have given explicit expressions for the probability that the typical user is covered by exactly $m$ BSs in the $\mathrm{SINR}$ model without frequency reuse. For general shadowing, they have calculated the probability 
\begin{align}
\label{pSINR}
p_m^{\mathrm{SINR}} &:=  \mathbf{P}\left[\mathcal{N}\left(T\right)=m\right]\\
&=\sum_{n=k}^{\infty}  (-1)^{n-k}{n\choose k}\mathcal{S}_n(T)\,,\nonumber
\end{align}
where  
$$\mathcal{S}_n(T)=\textstyle{\left(\frac{T}{1-(n-1)T}\right)^{-2n/\beta}}  \mathcal{I}_{n,\beta}(Wa^{-\beta/2})  \mathcal{J}_{n,\beta}\textstyle{\left(\frac{T}{1-(n-1)T}\right)}$$
for $0<T<1/(n-1)$  and $\mathcal{S}_n(T)=0$ otherwise, 
with $a=\lambda\pi \mathbb{E}[S^{\frac{2}{\beta}}]/B^{2}$, $\lambda$ is the intensity of base stations and 
\begin{equation}\label{In}
\mathcal{I}_{n,\beta}(x)=\frac{2^n
\int_0^{\infty} u^{2n-1}e^{-u^2-u^\beta x\Gamma(1-2/\beta)^{-\beta/2}} du
}{\beta^{n-1}(C'(\beta))^n(n-1)!}
\end{equation}
where 
\begin{equation}
 C'(\beta)=\frac{2\pi}{\beta\sin(2\pi/\beta)}=
\Gamma(1-2/\beta)\Gamma(1+2/\beta).
\end{equation}
\begin{equation}\label{e.Jn}
 \mathcal{J}_{n,\beta}(x)=\int_{[0,1]^{n-1}}
 \frac{   \prod\limits_{i=1}^{n-1}   v_i^{i(2/\beta+1)-1}(1-v_i)^{2/\beta}
  }{\prod\limits_{i=1}^{n-1} (x+\eta_i)} dv_1\dots
dv_{n-1}
\end{equation}
where 
$\eta_i:=(1-v_i)\prod_{k=i+1}^{n-1}v_k$.
The software developed for MATLAB is available in \cite{KeelerMATLABk} to get the numerical values of $p_m^{\mathrm{SINR}}$. 

In the interference limited network $(W=0)$ some results regarding the Poisson-Dirichlet model can be used to calculate equivalently the above coverage probabilities, cf~\cite[Prop,~6]{sinrPD}. 




\subsubsection{Boolean Model} For the noise-limited case, where the interference is small compared to noise, we can use the Boolean model to calculate the probability of coverage by $m$ BSs. This is a germ-grain model, where the atoms of the PPP are the germs. Centered on each atom is a grain, i.e. a 2D sphere $\mathcal{B}\left(x_i,R_b\right)$ which describes the area of coverage. $R_b$ is a fixed radius that can be expressed by communications quantities. Specifically, if we only consider path-loss and no fading, the received signal at the boundary should be larger than the threshold, in order to guarantee coverage, i.e. $(\tilde{B}R_b)^{-\beta}\geq T$ $\Rightarrow$ $R_b=T^{-1/\beta}\tilde{B}^{-1}$.
It is shown in \cite[Lemma 3.1]{BacBlaVol1} that the number of BSs covering the typical user follows a Poisson distribution with parameter $\nu=\lambda\pi\left(T^{-1/\beta}\tilde{B}^{-1}\right)^2$ and we get
\begin{equation}
\label{pBoo}
p_m^B  =  \frac{\nu^m}{m!}e^{-\nu}.
\end{equation}


\subsubsection{Overlaid 2-Network Model} Very often in practice, it occurs that 2 (or more) networks of the same provider operate in parallel over an area, using different infrastructure (nodes) and orthogonal resources (bandwidth). It is typical, for example, for operators to have one network of base stations for 3G/4G technology and numerous WiFi hotspots within a city. Given that a user may chose between the two to connect to the Internet with his/her cellphone (and assuming for simplicity mathematically independent models of these two networks), the coverage number at the typical user is distributed as the convolution of the coverage probability vectors of the two individual networks $\mathbf{p}^{(1)}=\left[p_1^{(1)},\ldots,p_M^{(1)}\right]$ and $\mathbf{p}^{(2)}=\left[p_1^{(2)},\ldots,p_M^{(2)}\right]$, that is (with $p_m^{(\cdot)}:=0$ for $m<0$)

\begin{equation}
\label{2Net}
\mathbf{p}^{2NET} = \mathbf{p}^{(1)}*\mathbf{p}^{(2)} \ \Rightarrow \  p_m^{2NET} = \sum_{n=0}^M p_n^{(1)}p_{m-n}^{(2)}.
\end{equation}

\subsection{Performance of the Content Placement Policies}
In this section 
we show the performance benefits of our scheme compared to a standard policy, the one that places in the cache memory of size $K$, always $K$ Most Popular Contents [MPC]. For the [MPC] policy, $b_1=\ldots=b_K=1$ $\&$ $b_{K+1}=\ldots=b_J=0$ and the objective function is always equal to $f^{(MPC)}=(1-p_0)\sum_{j=1}^K a_j$. 
As shown in the following plots, when the user has significant probability to access more than one cache, the [MPC] is suboptimal. This result is intuitive, because a user covered by $m>1$ BSs, can search in $mK$ memory slots instead of $K$. 

 \begin{figure}[t!]    
\centering  
\label{CacheEval}
\subfigure{          
           \includegraphics[trim = 7mm 30mm 10mm 25mm, clip, width=0.3\textwidth]{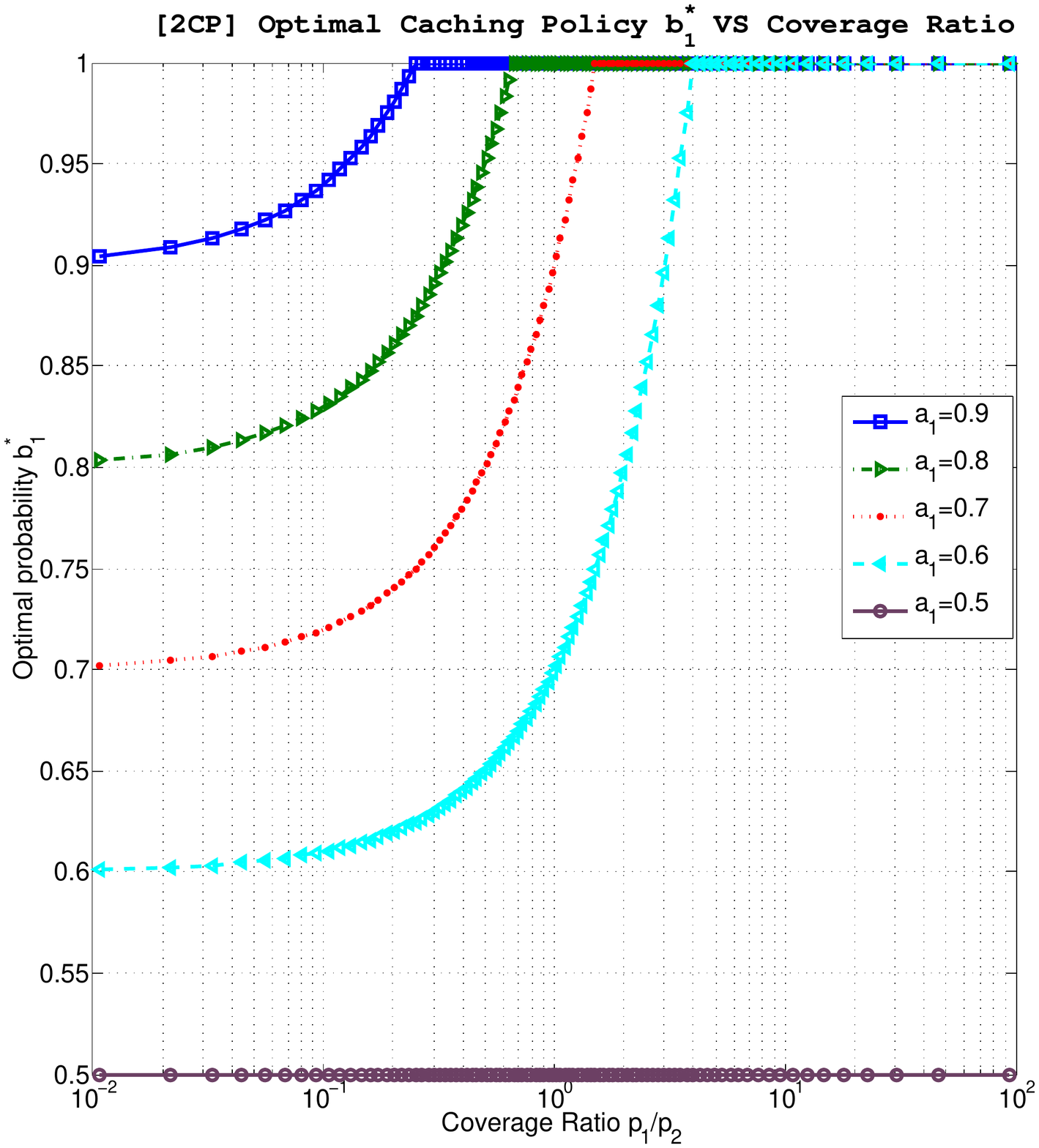}
           \label{CacheEval:2CPb}
           }
	   \subfigure{          
           \includegraphics[trim = 7mm 30mm 8mm 25mm, clip, width=0.3\textwidth]{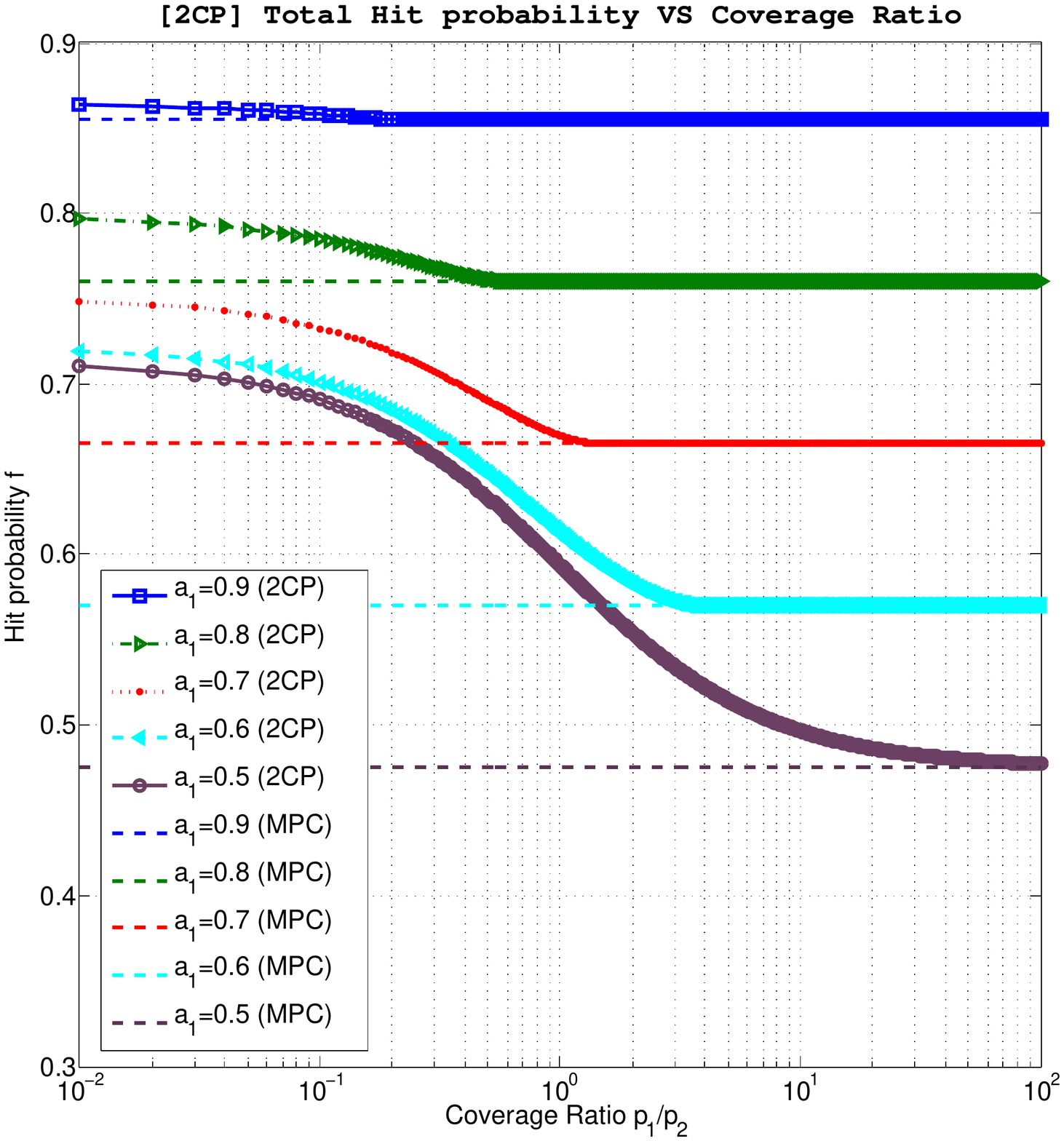}
           \label{CacheEval:2CPf}
           }
\caption{Case [2PC]: (a) The optimal caching policy $b_1^*$ and (b) The maximum hit probability $f^*$ (objective function), with respect to the coverage ratio $p_1/p_2$. The evaluation is done for different values of $a_1$. In (b) the optimal hit probability value is compared with the one when [MPC] policy is applied.}
\end{figure}

\subsubsection{Simple Scenario [2CP]}
We start by the solution of the [2CP] provided in Example~\ref{Cor1}.  In the simulated example, we assume $p_0=0.05$ to be the probability that the typical user is not covered by any BS. Hence $p_1+p_2=0.95$. A general picture of the way the optimal caching policy $\left(b_1^*,b_2^*\right)$ varies w.r.t. the coverage ratio $p_1/p_2$ and for different values of the popularity $a_1$, is given in Fig. \ref{CacheEval:2CPb}. Here, the ratio $p_1/p_2$ varies from $10^{-2}\rightarrow 10^2$ and we find for each value the optimal $b_1^*$, given $a_1=\left\{0.5,0.6,0.7,0.8,0.9\right\}$. 
We can deduce from the figure, that when each location is covered with high probability by 2 BSs, it is optimal to cache with probability $b_1^*\approx a_1$. When each location is covered with high probability by a single BS, it is optimal to cache with [MPC], 
i.e. $b_1^*\approx 1$.

In Fig. \ref{CacheEval:2CPf} we plot the objective function $f^*$ of [2CP], given the solution $(b_1^*,b_2^*)$ and for different values of $a_1$. We compare the solution to the value of the objective function under the [MPC], which is always equal to $f^{(MPC)}=a_1(1-p_0)$, irrespective of the values of $p_1,p_2$. From the figure we can observe a considerable performance improvement in the total hit probability, which is especially large when $a_1/a_2$ is small (comparable popularities) and when $p_1/p_2$ is small. %

\begin{figure*}[th!]    
\centering  
\label{CacheEval2}        
           \subfigure[Hit probability (Boolean).]{          
           \includegraphics[trim = 5mm 30mm 10mm 30mm, clip, width=0.3\textwidth]{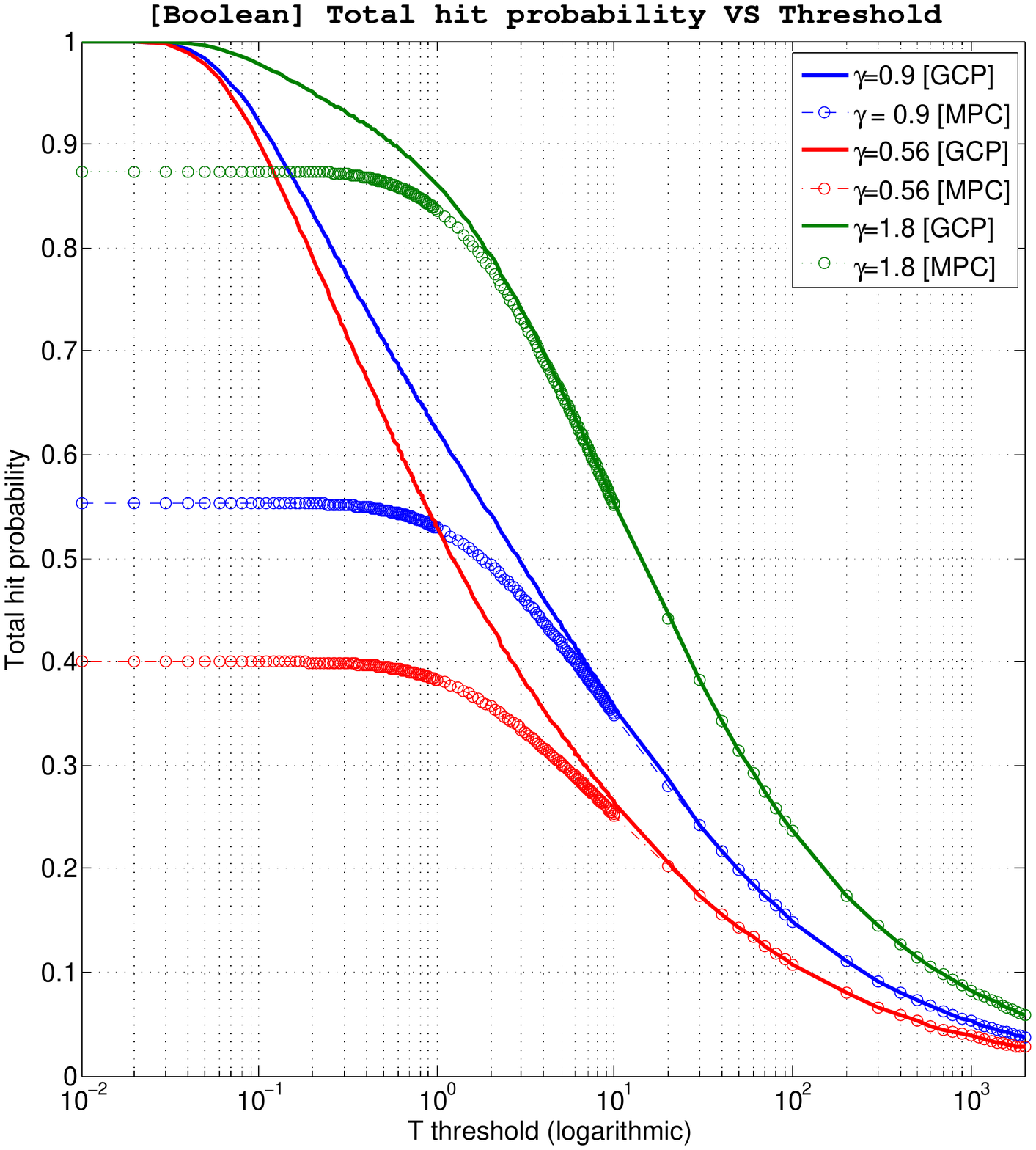}
           \label{CacheEval2:Bool}
           }
	   \subfigure[Hit probability ($\mathrm{SINR}$).]{          
           \includegraphics[trim = 7mm 30mm 10mm 30mm, clip, width=0.3\textwidth]{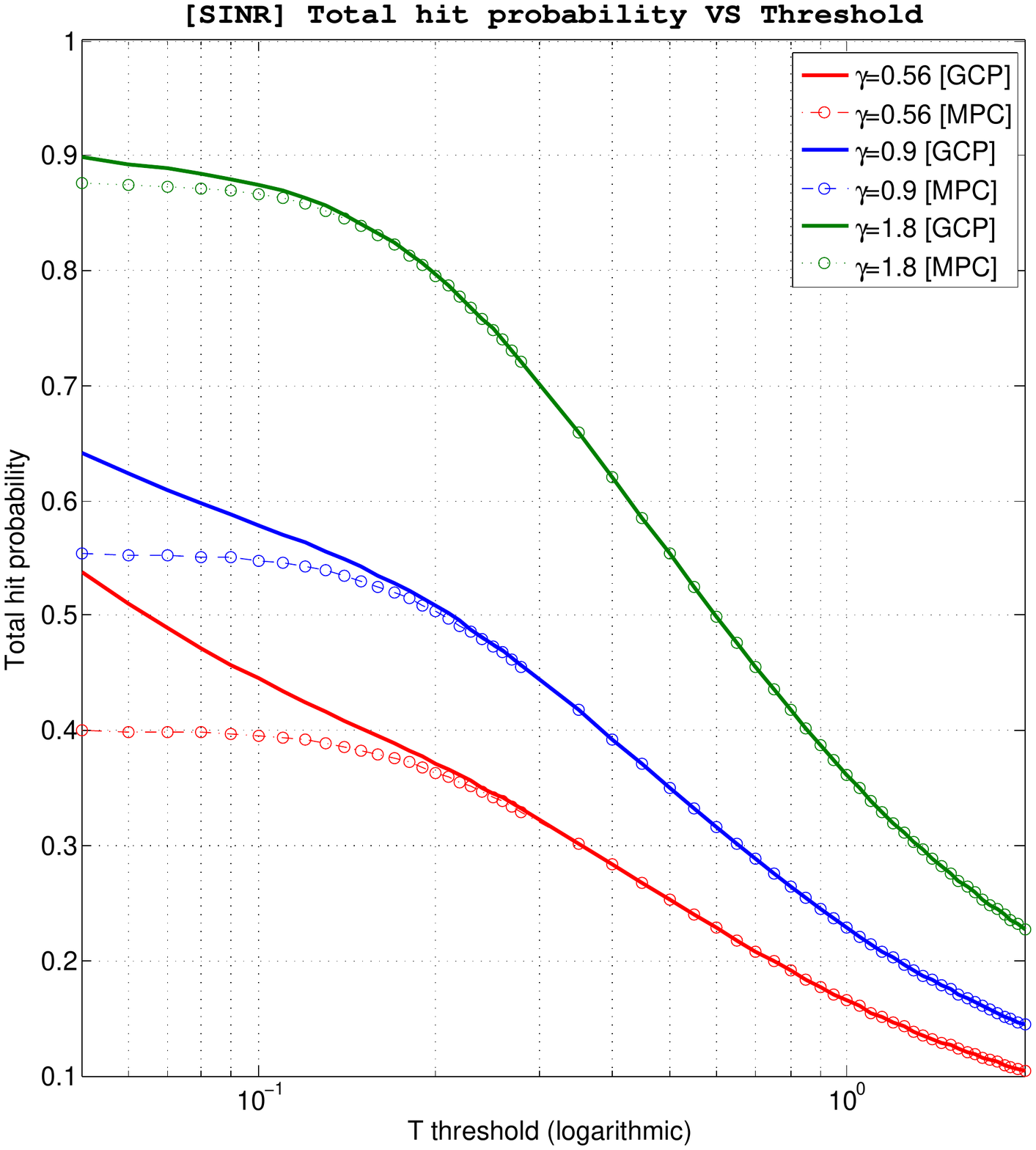}
           \label{CacheEval2:SINR}
           }
	   \subfigure[Hit probability (2NET).]{          
           \includegraphics[trim = 7mm 30mm 10mm 30mm, clip, width=0.3\textwidth]{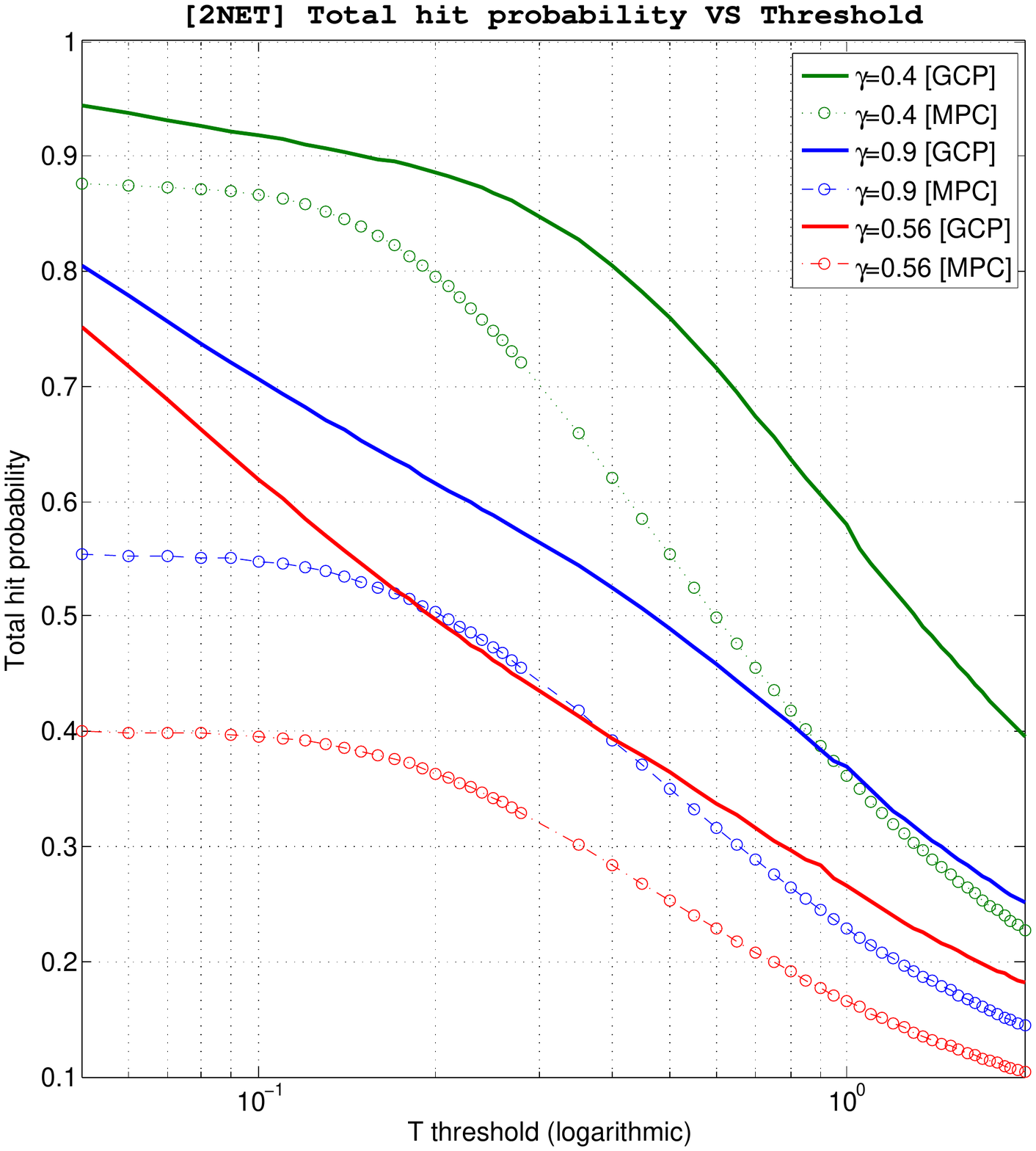}
           \label{CacheEval2:2NET}
           }
\caption{Evaluation of the optimal policy [GCP] and comparison with the [MPC] policy for the three different coverage models.}
\end{figure*}

\subsubsection{Boolean, $\mathrm{SINR}$ and Overlaid 2-Network Coverage}

We further evaluate the general problem [GCP] for the three coverage models suggested in Section \ref{secII}. We consider a content library of size $J=25$ and cache memories of size $K=5$. Since in all three models coverage depends on the threshold ratio $T$, we use the latter as the variable on the x-axis. In all cases, increasing the service threshold $T$ reduces the probability of coverage by $m>0$ BSs, and consequently increases $p_0$. Another important aspect is the relation of $T$ with the transmission rate $R$. These two are related through the Shannon formula 
$R = B_W\frac{1}{2}\log_2\left(1+T\right)$,
where the transmission bandwidth is considered here equal to $B_W=5$ MHz for typical applications.

In Fig. \ref{CacheEval2:Bool} we evaluate the total hit probability under the Boolean model, for which the values of $p_m^B$ are calculated as in (\ref{pBoo}). 
We choose $M=10$. The evaluation spans the threshold values $T=10^{-2}\rightarrow 2\cdot 10^3$. Compared to [MPC], we observe considerable gains in hit probability until $T\leq 10$, which corresponds to rate service of $R\leq 8650$ Kbits/sec. Given that $8000$ Kbits/sec is a very high \textit{video} quality from YouTube \cite{YouTubeEncode}, our approach can realistically improve the backhaul cellular network  traffic under this model.

In Fig. \ref{CacheEval2:SINR} the same performance evaluation is done for the $\mathrm{SINR}$ model, where the probability of coverage is found in closed form in \cite{KeelerBartek13}. We use the software developed for MATLAB and available in \cite{KeelerMATLABk} to get the numerical values of $p_m^{\mathrm{SINR}}$. These are used as input to solve the [GCP]. We chose to evaluate the interference-limited case, i.e. $W=0$, thus we consider $\mathrm{SIR}$. For numerical integration reasons, the minimum threshold is taken to be $5\cdot 10^{-2}$, which from (\ref{SINRT}) refers to at most $25$ BSs covering a planar point. The maximum threshold value is $2$ because due to (\ref{SINRT}) at most $1$ BS can cover a planar point when $T\geq 1$. From the figure we see that the benefits are not very important and appear until $T\leq 0.2$, or equivalently $R\leq 650$ Kbits/sec. This rate refers to \textit{audio files} rather than video files, given that the a high quality encoded audio file has a rate of $512$ Kbits/sec. 
The main reason for the poor performance is the generally low probability of coverage by more than one BS (around $20\%$ at best). We conclude that in the $\mathrm{SINR}$ model without frequency reuse, it is optimal to use [GCP] for low bit rate content (audio) and [MPC] for high bit rate content (video).

Finally, Fig. \ref{CacheEval2:2NET} illustrates the performance gains when the [GCP] is applied to the case of coverage by 2 independent overlaid networks (2NET). The coverage probability $p_m^{2NET}$ is given in (\ref{2Net}). 
For both vectors of the convolution, we use the same numerical values from $p_m^{\mathrm{SINR}}$ as calculated in the single $\mathrm{SINR}$ network above. Due to the convolution, the coverage probability for $m>1$ is now increased, since most planar areas will be covered by at least two BSs. In such case the [GCP] policy has impressive benefits in the entire domain of $T$, compared to the [MPC]. More than any other, this case emphasises the great potentials of optimal geographic caching of content.

\section{Conclusions}
In this work, we have revisited the problem of optimal content placement in caches within a cellular network. We exploited the fact that certain areas are covered by multiple BSs. An optimal policy is derived which suggests that when multi-coverage areas are significant, it is not optimal to cache the most popular contents everywhere. The total hit probability of the policy is evaluated in plots for three different coverage models (Boolean, $\mathrm{SINR}$, Overlaid 2-Network) and the results are highly in favour of our approach.

\label{secVI}

\bibliographystyle{unsrt}
\footnotesize
\bibliography{CacheStoGeoB.bib}

\begin{thebibliography}{10}

\bibitem{BacGiov13}
F.~Baccelli and A.~Giovanidis.
\newblock Coverage by pairwise base station cooperation under adaptive
  geometric policies.
\newblock {\em Proc. of 47th Asilomar Conference on Signals, Systems and
  Computers}, 2013.

\bibitem{PoularINFOCOM14}
K.~Poularakis, G.~Iosifidis, A.~Argyriou, and L.~Tassiulas.
\newblock Video delivery over heterogeneous cellular networks: Optimizing cost
  and performance.
\newblock {\em INFOCOM, Toronto, Canada}, 2014.

\bibitem{DimakisCaireCacheIT}
K.~Shanmugam, N.~Golrezaei, A.~G. Dimakis, A.~F. Molisch, and G.~Caire.
\newblock Femto{C}aching: Wireless content delivery through distributed caching
  helpers.
\newblock {\em IEEE Trans. IT, Vol:59, Iss: 12}, 2013.

\bibitem{BastugCacheISWCS14}
E.~Bastug, M.~Bennis, and M.~Debbah.
\newblock Cache-enabled {S}mall {C}ell {N}etworks: {M}odeling and {T}radeoffs.
\newblock {\em 11th ISWCS, Barcelona, Spain}, Aug. 2014.

\bibitem{PoularISIT13}
K.~Poularakis and L.~Tassiulas.
\newblock Exploiting user mobility for wireless content delivery.
\newblock {\em ISIT, Istanbul, Turkey}, 2013.

\bibitem{Wang5Gcache}
X.~Wang, M.~Chen, T.~Taleb, A.~Ksentini, and V.~Leung.
\newblock Cache in the {A}ir: {E}xploiting content caching and delivery
  techniques for 5{G} systems.
\newblock {\em IEEE Comm. Mag., Vol:52, Iss:2}, 2014.

\bibitem{Newman05}
M.~E.~J. Newman.
\newblock Power laws, {P}areto distributions and {Z}ipf's law.
\newblock {\em Contemporary Physics 46, pp. 323Ð351}, 2005.

\bibitem{KeelerBartek13}
H.~P. Keeler, B.~B\l aszczyszyn, and M.~K. Karray.
\newblock {SINR}-based k-coverage probability in cellular networks with
  arbitrary shadowing.
\newblock In {\em In Proc. of IEEE ISIT}, 2013.

\bibitem{LeonardiCacheUni}
V.~Martina, M.~Garetto, and E.~Leonardi.
\newblock A unified approach to the performance analysis of caching systems,.
\newblock {\em Proc. INFOCOM, Toronto, Canada}, 2014.

\bibitem{BacBlaVol1}
F.~Baccelli and B.~B{\l}aszczyszyn.
\newblock {\em Stochastic Geometry and Wireless Networks, Volume I --- Theory},
  volume 3, No 3--4 of {\em Foundations and Trends in Networking}.
\newblock NoW Publishers, 2009.

\bibitem{BoydBook}
S.~Boyd and L.~Vandenberghe.
\newblock {\em Convex Optimization}.
\newblock Cambridge University Press, 2004.

\bibitem{KeelerMATLABk}
H.~P. Keeler.
\newblock {SINR}-based k-coverage probability in cellular networks.
\newblock MATLAB Central
  \url{http://www.mathworks.fr/matlabcentral/fileexchange/
  40087-sinr-based-k-coverage-probability-in-cellular-networks}.
\newblock accessd on 2014.09.19.

\bibitem{sinrPD}
H.~P. Keeler and B.~B{\l}aszczyszyn.
\newblock {SINR} in wireless networks and the two-parameter
  {P}oisson-{D}irichlet process.
\newblock {\em IEEE Wireless Comm Letters}, 2014.

\bibitem{YouTubeEncode}
YouTube.
\newblock Advanced encoding settings [online].

\end{thebibliography}

\end{document}